\documentclass[11pt]{article}%
\usepackage{amsfonts}
\usepackage{amsmath}
\usepackage{amssymb}
\usepackage{graphicx}%
\providecommand{\U}[1]{\protect\rule{.1in}{.1in}}
\newtheorem{theorem}{Theorem}
\newtheorem{acknowledgement}[theorem]{Acknowledgement}

\newtheorem{claim}[theorem]{Claim}

\newtheorem{corollary}[theorem]{Corollary}

\newtheorem{example}[theorem]{Example}

\newtheorem{lemma}[theorem]{Lemma}

\newtheorem{proposition}[theorem]{Proposition}

\newenvironment{proof}[1][Proof]{\noindent\textbf{#1.} }{\ \rule{0.5em}{0.5em}}
\begin{document}

\title{Polar Duality and the Reconstruction of Quantum Covariance Matrices from
Partial Data}
\author{Maurice de Gosson\\University of Vienna\\Institute of Mathematics \ (NuHAG)\\maurice.de.gosson@univie.ac.at}
\maketitle
\tableofcontents

\begin{abstract}
We address the problem of the reconstruction of quantum covariance matrices
using the notion of Lagrangian and symplectic polar duality introduced in
previous work. We apply our constructions to Gaussian quantum states which
leads to a non-trivial generalization of Pauli's reconstruction problem and we
state a simple tomographic characterization of such states.

\end{abstract}


\textbf{MSC classification 2020}: 52A20, 52A05, 81S10, 42B35

\section{Introduction}

The covariance matrix is a fundamental concept in both classical and quantum
mechanics, serving distinct purposes in each domain. In classical mechanics,
the covariance matrix is employed to characterize statistical relationships
and correlations between different variables within a system \cite{Diaz}. In
quantum mechanics, the covariance matrix holds particular significance in the
context of quantum correlations \cite{Eli}. According to Born's rule, the
quantum covariance matrix encapsulates all available statistical information
about a quantum state. Moreover, covariance matrices serve as a powerful tool
for detecting entanglement, playing a key role in identifying and analyzing
quantum entangled states \cite{Gitt,wolf}.

Notably, the quantum covariance matrix fully characterizes Gaussian states,
their Wigner distributions are parametrized by their covariance matrices and
their center. This correspondence is central to the discussion in the present
paper to address the problem of the determination of covariance matrices from
partial data. This problem is not new and has been studied by many authors,
see for instance \v{R}eh\'{a}\v{c}ek \textit{et al}. \cite{Rehacek}. We have
initiated such a study from a wider point of view in \cite{math,sym,GSI23}
using the notion of polar duality from convex geometry; the present work
considerably extends and synthesizes these preliminary works. More precisely,
let $\Sigma$ be an arbitrary real positive definite $2n\times2n$ matrix; such
a matrix can always be viewed as the covariance of the multivariate Gaussian
distribution
\begin{equation}
\rho(z)=\frac{1}{(2\pi)^{n}\sqrt{\det\Sigma}}e^{-\frac{1}{2}\Sigma
^{-1}(z-z_{0})\cdot(z-z_{0})} \label{Gauss0}%
\end{equation}
(we are using the notation $z=(x,p)\in\mathbb{R}_{x}^{n}\times\mathbb{R}%
_{p}^{n}$); this distribution qualifies as a the Wigner distribution of a
quantum state provided that we impose a constraint on $\Sigma$; the latter is
usually chosen to be \cite{dutta,Birk,golu09,Narco}:
\begin{equation}
\text{\textit{The eigenvalues of} }\Sigma+\frac{i\hbar}{2}J\text{ \textit{are
all}}\geq0 \label{Quantum0}%
\end{equation}
where $J$ is the standard symplectic matrix. We have shown
\cite{FOOP,blob,golu09} that this condition is equivalent to%
\begin{equation}
\text{\textit{The covariance ellipsoid} }\Omega_{\Sigma}\text{
\textit{contains a quantum blob}} \label{Quantum1}%
\end{equation}
(a quantum blob is a symplectic ball with radius $\sqrt{\hbar}$); in
\cite{FOOP,golu09} we also formulated these conditions using the topological
notion of symplectic capacity \cite{jphysa} which is closed related to
Gromov's symplectic non-squeezing theorem. On the operator level the
conditions (\ref{Quantum0}) and (\ref{Quantum1}) guarantees that the trace
class operator with Weyl symbol $(2\pi\hbar)^{n}\rho$ is positive semidefinite
and has trace one, and thus qualifies as a density operator \cite{QHA}.

The main results of this paper are

\begin{itemize}
\item Theorem \ref{ThmRadon} which states that a generalized Gaussian (and
hence a covariance matrix) can be reconstructed from the knowledge of only two
marginal distributions of along a pair of transverse Lagrangian planes; it
provides a generalization of the solution of Pauli's reconstruction problem;

\item Theorem \ref{ThmA} provides a geometric interpretation of the previous
results; identifies Gaussian states with quantum blobs viewed as John
ellipsoids of a convex set constructed using the notion of Lagrangian polar
duality, which can be viewed as a geometric variant of the uncertainty principle:

\item Theorem \ref{ThmB} uses the notion of symplectic polar duality
$\Omega\longrightarrow\Omega^{\hbar,\omega}$ we introduces in
\cite{FOOPMC,Bull}. We prove that in order to prove that a phase space
ellipsoid $\Omega$ is the covariance ellipsoid of a Gaussian quantum state it
suffices that $\Omega^{\hbar,\omega}\cap\ell\subset\Omega\cap\ell$ for one
Lagrangian subspace $\ell$ (and hence for all).
\end{itemize}

\paragraph{Notation and terminology}

We will denote by $\omega$ the standard symplectic form on $\mathbb{R}%
^{2n}\equiv\mathbb{R}_{x}^{n}\times\mathbb{R}_{p}^{n}$ that is, in matrix
form\textit{,} $\omega(z,z^{\prime})=Jz\cdot z^{\prime}$ where $z=(x,p)^{T}$,
$z^{\prime}=(x^{\prime},p^{\prime})^{T}$ and $J=%
\begin{pmatrix}
0 & I\\
-I & 0
\end{pmatrix}
$. The scalar product of two vectors $u,v\in\mathbb{R}^{m}$ is written $u\cdot
v$. We denote by $\operatorname*{Sym}_{++}(m,\mathbb{R})$ the convex set of
all symmetric positive definite real $m\times m$ matrices.

The group of all automorphisms of the symplectic space $(\mathbb{R}%
^{2n},\omega)$ leaving the standard symplectic form invariant is called the
(standard) symplectic group and is denoted by $\operatorname*{Sp}(n)$. The
properties of $\operatorname*{Sp}(n)$ and of its double covering, the
metaplectic group $\operatorname*{Mp}(n)$ are summarized in the Appendix A.

A linear subspace of $(\mathbb{R}^{2n},\omega)$ with dimension $n$ is on which
$\omega$ vanishes identically is called a Lagrangian subspace (or plane); the
set of all Lagrangian subspaces of $(\mathbb{R}^{2n},\omega)$ will be denoted
by $\operatorname*{Lag}(n)$, and is called the Lagrangian Grassmannian of
$(\mathbb{R}^{2n},\omega)$.

Let $B^{2n}(\sqrt{\hbar})$ be the ball with center $0$ and radius $\sqrt
{\hbar}$ in $\mathbb{R}^{2n}$ (equipped with the usual Euclidean norm). The
image $S(B^{2n}(\sqrt{\hbar}))$ of that ball by $S\in\operatorname*{Sp}(n)$ is
called a \textit{quantum blob. }

\section{The RSUP for Gaussians}

The Robertson--Schr\"{o}dinger uncertainty principle (RSUP) is, as opposed to
the elementary\ Heisenberg inequalities, a fundamental concept in quantum
mechanics that describe the trade-off between the uncertainties in the
measurements of two non-commuting observables, such as position and momentum.
These inequalities are typically expressed in terms of standard deviations or (co)variances.

\subsection{Multivariate Correlated Gaussians\label{secGauss}}

We denote by $\phi_{0}$ the standard Gaussian (called \textquotedblleft
fiducial state\textquotedblright\ by Littlejohn \cite{Littlejohn}) and by
$W\phi_{0}$ its Wigner function:%
\begin{equation}
\phi_{0}(x)=(\pi\hbar)^{-n/4}e^{-|x|^{2}/2\hbar}\text{ \ },\text{ \ }W\phi
_{0}(x)=(\pi\hbar)^{-n}e^{-|z|^{2}/\hbar}; \label{wigfi}%
\end{equation}
here $|x|^{2}=x_{1}^{2}\div\cdot\cdot\cdot\div+x_{n}^{2}$ and $|z|^{2}%
=|x|^{2}+|p|^{2}$. Consider now the set $\operatorname*{Gauss}_{0}(n)$ of all
centered generalized Gaussians: we have $\psi_{XY}\in\operatorname*{Gauss}%
_{0}(n)$ if and only if
\begin{equation}
\psi_{X,Y}^{\gamma}(x)=i^{\gamma}\left(  \tfrac{1}{\pi\hbar}\right)
^{n/4}(\det X)^{1/4}e^{-\tfrac{1}{2\hbar}(X+iY)x\cdot x} \label{psixy}%
\end{equation}
where $X$ and $Y$ are symmetric real $n\times n$ matrices such that $X>0$ and
$\gamma\in\mathbb{R}$ is a real constant. This function is normalized to
unity: $||\psi_{X,Y}^{\gamma}||_{L^{2}}=1$ and its Wigner transform is given
by \cite{Birk,WIGNER}
\begin{equation}
W\psi_{X,YX,Y}^{\gamma}(z)=\left(  \tfrac{1}{\pi\hbar}\right)  ^{n}%
e^{-\tfrac{1}{\hbar}Gz^{2}} \label{wxy}%
\end{equation}
where $G$ is the symmetric and symplectic matrix
\begin{equation}
G=%
\begin{pmatrix}
X+YX^{-1}Y & YX^{-1}\\
X^{-1}Y & X^{-1}%
\end{pmatrix}
~. \label{gsym}%
\end{equation}
(The result seems to back to Bastiaans, see \cite{Littlejohn}.) Observe that
$G=S^{T}S$ where
\begin{equation}
S=%
\begin{pmatrix}
X^{1/2} & 0\\
X^{-1/2}Y & X^{-1/2}%
\end{pmatrix}
\in\operatorname*{Sp}(n). \label{bi}%
\end{equation}
An almost trivial, but fundamental, observation is that every $\psi
_{X,Y}^{\gamma}\in\operatorname*{Gauss}_{0}(n)$ can be obtained from the
standard Gaussian using elementary metaplectic transformations, namely%
\begin{equation}
\psi_{X,Y}^{\gamma}=i^{\gamma}\widehat{V}_{Y}\widehat{M}_{X^{1/2},0}\phi_{0}.
\label{psifi}%
\end{equation}
In fact:

\begin{proposition}
\label{PropGauss}The metaplectic group $\operatorname*{Mp}(n)$ acts
transitively on Gaussians $\psi_{X,Y}=\psi_{X,Y}^{0}$ up to a unimodular
factor:
\[
\operatorname*{Mp}(n)\times\operatorname*{Gauss}\nolimits_{0}%
(n)\longrightarrow\operatorname*{Gauss}\nolimits_{0}(n).
\]

\end{proposition}

\begin{proof}
To see this it is sufficient, in view of formula (\ref{psifi}) to show that if
$\widehat{S}\in\operatorname*{Mp}(n)$ then $\widehat{S}\phi_{0}=i^{\mu}%
\psi_{X,Y}$ for some $\mu\in\mathbb{R}$ and\ $\psi_{X,Y}\in
\operatorname*{Gauss}\nolimits_{0}(n)$. In view of the pre-Iwasawa
factorization (Appendix A, formula (\ref{iwa1})) we have $S=\pi
^{\operatorname*{Mp}}(\widehat{S})=V_{P}M_{L}R$ and hence%
\begin{equation}
\widehat{S}=\pm\widehat{V}_{P}\widehat{M}_{L,0}\widehat{R}. \label{iwa2}%
\end{equation}
We claim that It follows that $\widehat{R}\phi_{0}=i^{\mu}\phi_{0}$ for some
$\mu\in\mathbb{R}$. In fact we have, by the symplectic covariance of the
Wigner transform,
\[
W(\widehat{R}\phi_{0})(z)=W\phi_{0}(R^{-1}z)=W\phi_{0}(z)
\]
the second equality in view of the rotational invariance of $\phi_{0}$; it
follows that
\[
\widehat{S}\phi_{0}=\pm i^{\mu}\widehat{V}_{P}\widehat{M}_{L,0}\widehat{R}%
\phi_{0}=\pm i^{\mu}\psi_{L^{2},P}%
\]
hence $\widehat{S}\phi_{0}\in\operatorname*{Gauss}\nolimits_{0}(n)$ as claimed
(this result can also be obtained using Fourier integral operators, at the
cost of much more complicated calculations involving Cayley transforms (see
\textit{e.g}. \cite{Littlejohn}). There remains to show that $\widehat{S}%
\psi_{X,Y}^{\gamma}\in\operatorname*{Gauss}\nolimits_{0}(n)$ for all
$\psi_{X,Y}^{\gamma}\in\operatorname*{Gauss}\nolimits_{0}(n)$. Every
$\widehat{S}\in\operatorname*{Mp}(n)$ can be written $\widehat{S}%
=\widehat{V}_{P}\widehat{M}_{L,m}\widehat{R}$ ($m\in\{0,2\}$) corresponding to
the pre-Iwasaw factorization of $S=\pi^{\operatorname*{Mp}}(\widehat{S})$. For
$\widehat{S^{\prime}}\in\operatorname*{Mp}(n)$; we have
\[
\widehat{S^{\prime}}\psi_{X,Y}=\widehat{S^{\prime}}\psi\widehat{V}%
_{Y}\widehat{M}_{X^{1/2},0}\phi_{0}=\widehat{S}\phi_{0}=\widehat{V}%
_{P}\widehat{M}_{L,m}\widehat{R}\phi_{0}.
\]
We claim that $\widehat{R}\phi_{0}=i^{\mu}\phi_{0}$ for some $\mu\in
\mathbb{R}$. We have
\[
W(\widehat{R}\phi_{0})(z)=W\phi_{0}(R^{-1}z)=W\phi_{0}(z)
\]
the second equality in view of the rotational invariance of $\phi_{0}$; hence
\[
\widehat{S^{\prime}}\psi_{X,Y}=\widehat{S}\phi_{0}=i^{\mu}\widehat{V}%
_{P}\widehat{M}_{L,m}\widehat{R}\phi_{0}=i^{\mu}\psi_{L^{2},P}\in
\operatorname*{Gauss}\nolimits_{0}(n).
\]

\end{proof}

\subsection{The multivariate RSUP}

A mixed Gaussian quantum state centered at $z_{0}$ is a density operator
$\widehat{\rho}$ on $L^{2}(\mathbb{R}^{n})$ whose Wigner distribution $\rho$
is of the type%
\begin{equation}
\rho(z)=\frac{1}{(2\pi)^{n}\sqrt{\det\Sigma}}e^{-\frac{1}{2}\Sigma
^{-1}(z-z_{0})\cdot(z-z_{0})} \label{rhoG}%
\end{equation}
where $\Sigma$ is the covariance (or noise) matrix; the condition
$\widehat{\rho}\geq0$ is equivalent to saying that the eigenvalues of the
Hermitian matrix $\Sigma+\frac{i\hbar}{2}J$ are $\geq0$ \cite{dutta,Birk},
which we write for short as the L\"{o}wner ordering:
\begin{equation}
\Sigma+\frac{i\hbar}{2}J\geq0. \label{quant0}%
\end{equation}
The state $\widehat{\rho}$ is pure if and only if $\det\Sigma=\left(
\frac{\hbar}{2}\right)  ^{2n}$ (see below). The condition (\ref{quant0}) is
equivalent to the following three statements:

\begin{claim}
The symplectic eigenvalues $\lambda_{1}^{\omega},...,\lambda_{n}^{\omega}$
\ of $\Sigma$ are all $\geq\frac{1}{2}\hbar$;
\end{claim}

\begin{claim}
The covariance ellipsoid $\Omega_{\Sigma}=\{z:\frac{1}{2}\Sigma^{-1}z\cdot
z<1\}$ contains a quantum blob $Q=S(B^{2n}(\sqrt{\hbar}))$ ($S\in
\operatorname*{Sp}(n)$);
\end{claim}

\begin{claim}
The symplectic capacity of the covariance ellipsoid is at least $\pi\hbar$.
\end{claim}

Recall \cite{Birk} that the symplectic eigenvalues of a positive definite
matrix $\Sigma$ are the numbers $\lambda_{j}^{\omega}>0$, ($0\leq j\leq n$)
such that the $\pm i\lambda_{j}^{\omega}$ are the eigenvalues of $J\Sigma$
(which are the same as those of the antisymmetric matrix $\Sigma^{1/2}%
J\Sigma^{1/2}$). The terminology \textquotedblleft quantum
blob\textquotedblright\ to denote a symplectic ball with radius $\sqrt{\hbar}$
was introduced in \cite{blob,FOOP}). For the notion of symplectic capacity and
its applications to the uncertainty principle see \cite{FOOP}. The proof of
the first claim is well-known and widespread in the literature, see for
instance \cite{dutta,blob,golu09}. It is easily proven using Williamson's
symplectic diagonalization theorem \cite{Birk} which says that for every
$\Sigma>0$ there exists $S\in\operatorname*{Sp}(n)$ such that%
\begin{equation}
\Sigma=S^{T}DS\text{ \ , \ }D=%
\begin{pmatrix}
\Lambda^{\omega} & 0_{n\times n}\\
0_{n\times n} & \Lambda^{\omega}%
\end{pmatrix}
\label{Williamson}%
\end{equation}
where $\Lambda^{\omega}=\operatorname*{diag}(\lambda_{1}^{\omega}%
,...,\lambda_{n}^{\omega})$. The proof of the second claim easily follows from
a geometric argument; see \cite{FOOP,golu09}. The case where $\lambda
_{j}^{\omega}=\frac{1}{2}\hbar$ for all $j$ is of particular interest; in this
case we have $\Sigma=\frac{\hbar}{2}S^{T}S$ (and hence $\det\Sigma=\left(
\frac{\hbar}{2}\right)  ^{2n}$) and so that
\[
\rho(z)=\frac{1}{(\pi\hbar)^{n}}e^{-\frac{1}{2}S^{-1}(z-z_{0})\cdot
S^{-1}(z-z_{0})}=W\phi_{0}(S^{-1}(z-z_{0}))
\]
that is%
\[
\rho(z)=W(\widehat{S}\phi_{0})(z-z_{0})=W(\widehat{T}(z_{0})\widehat{S}%
\phi_{0})(z)
\]
\ where $\widehat{S}\in\operatorname*{Mp}(n)$ covers $S$ and $\widehat{T}%
(z_{0})$ is the Heisenberg--Weyl displacement operator \cite{Birk}:%
\[
\widehat{T}(z_{0})\widehat{S}\psi(x)=e^{\frac{i}{\hbar}(p_{0}\cdot x-\frac
{1}{2}p_{0}\cdot x_{0})}\psi(x-x_{0}).
\]
From these considerations follows that the Wigner distribution of the Gaussian
state $\widehat{\rho}$ is the Wigner function of a multivariate Gaussian
$\widehat{T}(z_{0})\psi_{X,Y}$ centered at $z_{0}$. We will write the
covariance matrix in block-form
\begin{equation}
\Sigma=%
\begin{pmatrix}
\Sigma_{XX} & \Sigma_{XP}\\
\Sigma_{PX} & \Sigma_{PP}%
\end{pmatrix}
\text{ \ },\text{ \ }\Sigma_{PX}=\Sigma_{XP}^{T} \label{sigma1}%
\end{equation}
with $\Sigma_{XX}=(\sigma_{x_{j}x_{k}})_{1\leq j,k\leq n}$, $\Sigma
_{PP}=(\sigma_{p_{j}p_{k}})_{1\leq j,k\leq n}$, $\Sigma_{PP}=(\sigma
_{x_{j}p_{k}})_{1\leq j,k\leq n}$.

\begin{proposition}
\label{Thm1}Let $\Sigma$ be the covariance matrix of the Gaussian state
$\widehat{\rho}$; (i) This state is pure, that is there exists $\psi_{X,Y}$
such that $\rho=\psi_{X,Y}$, if and only if
\begin{gather}
\Sigma_{XX}\Sigma_{PP}-(\Sigma_{XP})^{2}=\frac{1}{4}\hbar^{2}I_{n\times
n}\label{RS1}\\
\Sigma_{XX}\Sigma_{XP}=\Sigma_{PX}\Sigma_{XX}\text{ },\text{ }\Sigma
_{PX}\Sigma_{PP}=\Sigma_{PP}\Sigma_{XP}. \label{RS2}%
\end{gather}
(ii) We have $(\Sigma_{XP})^{2}\geq0$, i.e. the eigenvalues of $(\Sigma
_{XP})^{2}$ are $\geq0$ hence the generalizes Heisenberg inequality
$\Sigma_{XX}\Sigma_{PP}\geq\frac{\hbar^{2}}{4}I_{n\times n}$ holds.
\end{proposition}

\begin{proof}
Suppose that $\widehat{\rho}$ is a pure state; then $\Sigma=\frac{\hbar}%
{2}S^{T}S$ and the conditions (\ref{RS1})--(\ref{RS2}) are just a restatement
of the relations (\ref{cond1}) in APPENDIX\ A, taking into account the fact
that $S^{T}S>0$. To prove that $(\Sigma_{XP})^{2}\geq0$ we note that since
$\Sigma_{XX}\Sigma_{XP}=\Sigma_{PX}\Sigma_{XX}$ we have $\Sigma_{XP}%
=\Sigma_{XX}^{-1}\Sigma_{PX}\Sigma_{XX}$ hence $\Sigma_{XP}$ and $\Sigma_{PX}$
have the same eigenvalues; since $\Sigma_{PX}=\Sigma_{XP}^{T}$ these
eigenvalues must be real, hence those of $(\Sigma_{XP})^{2}$ are $\geq0$.
\end{proof}

In particular, when $n=1$ one recovers the usual saturated
Robertson--Schr\"{o}dinger uncertainty principle
\[
\sigma_{xx}\sigma_{pp}-\sigma_{xp}^{2}=\frac{\hbar^{2}}{4}%
\]
satisfied by all pure one-dimensional Gaussian states.

Note that when $\Sigma$ is diagonal, \textit{i.e.} the state is a tensor
product of one-dimensional functions one recovers the well-known fact that the
Heisenberg uncertainty inequalities are saturated (reduce to equalities):
$\sigma_{x_{j}xj}\sigma_{p_{j}p_{j}}=\hbar^{2}/4$ for all $j=1,...,n$.

\subsection{Orthogonal projections of the covariance ellipsoid}

We begin with a general result. For $M\in\operatorname*{Sym}_{++}%
(2n,\mathbb{R})$ we define the phase space ellipsoid%
\begin{equation}
\Omega=\{z\in\mathbb{R}^{2n}:Mz\cdot z\leq\hbar\}. \label{Mellipse}%
\end{equation}

Writing $M$ in block-matrix form $M=%
\begin{pmatrix}
M_{XX} & M_{XP}\\
M_{PX} & M_{PP}%
\end{pmatrix}
$ where the blocks are $n\times n$ matrices the condition $M\in
\operatorname*{Sym}_{++}(2n,\mathbb{R})$ ensures us \cite{zhang} that
$M_{XX}>0$, $M_{PP}>0$, and $M_{PX}=M_{XP}^{T}$ . Using classical formulas for
the inversion of block matrices \cite{Tzon} the inverse of $M$ is
\begin{equation}
M^{-1}=%
\begin{pmatrix}
(M/M_{PP})^{-1} & -(M/M_{PP})^{-1}M_{XP}M_{PP}^{-1}\\
-M_{PP}^{-1}M_{PX}(M/M_{PP})^{-1} & (M/M_{XX})^{-1}%
\end{pmatrix}
\label{Minverse}%
\end{equation}
where $M/M_{PP}$ and $M/M_{XX}$ are the Schur complements:%
\begin{align}
M/M_{PP}  &  =M_{XX}-M_{XP}M_{PP}^{-1}M_{PX}\label{schurm1}\\
M/M_{XX}  &  =M_{PP}-M_{PX}M_{XX}^{-1}M_{XP}. \label{schurm2}%
\end{align}
The following results is well-known (see for instance \cite{gopolar} and the
references therein):

\begin{lemma}
\label{LemmaProj}The orthogonal projections $\Pi_{\ell_{X}}\Omega$ and
$P=\Pi_{\ell_{P}}\Omega$ on the coordinate subspaces $\ell_{X}=\mathbb{R}%
_{x}^{n}\times0$ and $\ell_{P}=0\times\mathbb{R}_{p}^{n}$ of the ellipsoid
$\Omega$ are the ellipsoids%
\begin{align}
\Pi_{\ell_{X}}\Omega &  =\{x\in\mathbb{R}_{x}^{n}:(M/M_{PP})x\cdot x\leq
\hbar\}\label{boundb}\\
\Pi_{\ell_{P}}\Omega &  =\{p\in\mathbb{R}_{p}^{n}:(M/M_{XX})p\cdot p\leq
\hbar\}. \label{bounda}%
\end{align}

\end{lemma}

Consider now a Gaussian state (\ref{rhoG}) with covariance matrix\ $\Sigma=%
\begin{pmatrix}
\Sigma_{XX} & \Sigma_{XP}\\
\Sigma_{PX} & \Sigma_{PP}%
\end{pmatrix}
$; by definition the covariance (or Wigner) ellipsoid of this state is
\begin{equation}
\Omega_{\Sigma}=\{z:\tfrac{1}{2}\Sigma^{-1}z\cdot z\leq1\}. \label{omegas}%
\end{equation}
Setting $M=\frac{1}{2}\hbar\Sigma^{-1}$ the results above yield the inverse of
the covariance matrix $\Sigma$:
\begin{equation}
\Sigma^{-1}=%
\begin{pmatrix}
(\Sigma/\Sigma_{PP})^{-1} & -(\Sigma/\Sigma_{PP})^{-1}\Sigma_{XP}\Sigma
_{PP}^{-1}\\
-\Sigma_{PP}^{-1}\Sigma_{PX}(\Sigma/\Sigma_{PP})^{-1} & (\Sigma/\Sigma
_{XX})^{-1}%
\end{pmatrix}
. \label{covinv}%
\end{equation}
Formula (\ref{covinv}) immediately follows from (\ref{Minverse}) using the
relations
\begin{gather}
\Sigma_{XX}=\frac{\hbar}{2}(M/M_{PP})^{-1}\text{ },\text{ }\Sigma_{PP}%
=\frac{\hbar}{2}(M/M_{XX})^{-1}\label{msig1}\\
\Sigma_{XP}=-\frac{\hbar}{2}(M/M_{PP})^{-1}M_{XP}M_{PP}^{-1}. \label{msig2}%
\end{gather}

\begin{proposition}
\label{PropProj}The orthogonal projections on the canonical coordinate
subspaces $\ell_{X}$ and $\ell_{P}$ of $\Omega_{\Sigma}$ are
\begin{align}
\Pi_{\ell_{X}}\Omega_{\Sigma}  &  =\{x\in\mathbb{R}_{x}^{n}:\tfrac{1}{2}%
\Sigma_{XX}^{-1}x\cdot x\leq1\}\label{boundc}\\
\Pi_{\ell_{P}}\Omega_{\Sigma}  &  =\{p\in\mathbb{R}_{p}^{n}:\tfrac{1}{2}%
\Sigma_{PP}^{-1}p\cdot p\leq1\}. \label{boundd}%
\end{align}

\end{proposition}

\begin{proof}
The projection formulas (\ref{boundc}) and (\ref{boundd}) are a consequence of
the corresponding formulas (\ref{boundb}) and (\ref{bounda}).
\end{proof}

We next note the following remarkable fact showing that the notion of polar
duality is related to the uncertainty principle (see APPENDIX C for a short
review of polar duality):

\begin{proposition}
\label{Thm2}(i) The ellipsoids $X=\Pi_{\ell_{X}}\Omega_{\Sigma}$ and
$P=\Pi_{\ell_{P}}\Omega_{\Sigma}$ are such that $X^{\hbar}\subset P$ where
\begin{equation}
X^{\hslash}=\{p\in\mathbb{R}_{p}^{n}:\sup\nolimits_{x\in X}(p\cdot x)\leq
\hbar\} \label{omo2}%
\end{equation}
is the $\hbar$-polar dual of $X$. (ii) We have the equality $X^{\hbar}=P$ if
and only if $\Omega_{\Sigma}=B^{2n}(\sqrt{\hbar})$ in which case $X=B_{X}%
^{n}(\sqrt{\hbar})$ and $P=B_{P}^{n}(\sqrt{\hbar})$ (the centered balls with
radius $\sqrt{\hbar}$ in $\ell_{X}$ and $\ell_{P}$, respectively)
\end{proposition}

\begin{proof}
(i) Recall \cite{Bull,FOOPMC} that if.
\begin{equation}
X=\{x\in\mathbb{R}_{x}^{n}:Ax\cdot x\leq\hbar\} \label{Xell}%
\end{equation}
with $A\in\operatorname*{Sym}_{++}(n,\mathbb{R})$, then its $\hbar
$-polar$X^{\hslash}$ is the ellipsoid%
\begin{equation}
X^{\hslash}=\{p\in\mathbb{R}_{p}^{n}:A^{-1}p\cdot p\leq\hbar\}.
\label{Xelldual}%
\end{equation}
It follows that if
\begin{equation}
P=\{p\in\mathbb{R}_{p}^{n}:Bp\cdot p\leq\hbar\} \label{Polar}%
\end{equation}
(with $B\in\operatorname*{Sym}_{++}(n,\mathbb{R})$) the inclusion $X^{\hbar
}\subset P$ holds if and only if $AB\geq I_{n\times n}$ (\textit{i.e.} the
eigenvalues of $AB$ are all $\geq1$); this is equivalent to $BA\geq I_{n\times
n}$. Since we have here $A=\frac{1}{2}\hbar\Sigma_{XX}^{-1}$ and $B=\frac
{1}{2}\hbar\Sigma_{PP}^{-1}$ (formulas (\ref{boundc}) and (\ref{boundb})) the
inclusion $X^{\hbar}\subset P$ will hold if and only if $\frac{1}{4}\hbar
^{2}\Sigma_{XX}^{-1}\Sigma_{PP}^{-1}\geq I_{n\times n}$ or, equivalently
$\Sigma_{XX}\Sigma_{PP}\geq$ $\frac{1}{4}\hbar^{2}$. But this is the
generalized Heisenberg inequality of Proposition \ref{Thm1}. (ii) We have
$X^{\hbar}=P$ if and only if $AB=I_{n\times n}$.
\end{proof}

\begin{example}
\label{example1}Let us illustrate the result above in the case $n=1$. Here
$\Sigma_{XX}=\sigma_{xx}>0$, $\Sigma_{PP}=\sigma_{pp}>0$, and $\Sigma
_{XP}=\Sigma_{PX}=\sigma_{xp}$ and the covariance ellipse is
\begin{equation}
\Omega_{\Sigma}:\dfrac{\sigma_{pp}}{2D}x^{2}-\frac{\sigma_{xp}}{D}%
px+\dfrac{\sigma_{xx}}{2D}p^{2}\leq1 \label{2D}%
\end{equation}
where $D=\sigma_{xx}\sigma_{pp}-\sigma_{xp}^{2}\geq\tfrac{1}{4}\hbar^{2}$. The
orthogonal projections $\Omega_{X}$ and $\Omega_{P}$ of $\Omega$ on the $x$
and $p$ coordinate axes are the intervals
\begin{equation}
\Omega_{X}=[-\sqrt{2\sigma_{xx}},\sqrt{2\sigma_{xx}}]\text{\ \textit{\ , \ }
}\Omega_{P}=[-\sqrt{2\sigma_{pp}},\sqrt{2\sigma_{pp}}\dot{]}~.
\label{intervals}%
\end{equation}
Let $\Omega_{X}^{\hbar}$ be the polar dual of $\Omega_{X}$: it is the set of
all numbers $p$ such that $px\leq\hbar$\ for $-\sqrt{2\sigma_{xx}}\leq
x\leq\sqrt{2\sigma_{xx}}$and is thus the interval \
\[
\Omega_{X}^{\hbar}=[-\hbar/\sqrt{2\sigma_{xx}},\hbar/\sqrt{2\sigma_{xx}}]~.
\]
Since $\sigma_{xx}\sigma_{pp}\geq\frac{1}{2}\hbar$ we have the inclusion
\begin{equation}
\Omega_{X}^{\hbar}\subset\Omega_{P} \label{xompom}%
\end{equation}
which reduces to the equality $\Omega_{X}^{\hbar}=\Omega_{P}$ if and only if
the Heisenberg inequality is saturated, \textit{i.e.} $\sigma_{xx}\sigma
_{pp}=\frac{1}{4}\hbar^{2}$ which is equivalent to $\sigma_{xp}=0$.
\end{example}

\section{Pauli's Problem and its Generalizations\label{SecRec1}}

\subsection{Paulis' reconstruction problem}

The Pauli reconstruction problem is a particular case of a larger class of
phase retrieval problems; see Grohs and Liehr \cite{Grohs} for recent advances
on this difficult topic. Pauli asked in \cite{Pauli} whether the probability
densities $|\psi(x)|^{2}$ and $|\widehat{\psi}(p)|^{2}$ of a normed function
$\psi\in L^{2}(\mathbb{R})$ uniquely determine $\psi$. The answer is in
general negative: consider the correlated Gaussian
\begin{equation}
\psi(x)=\left(  \tfrac{1}{2\pi\sigma_{xx}}\right)  ^{1/4}e^{-\frac{x^{2}%
}{4\sigma_{xx}}}e^{\frac{i\sigma_{xp}}{2\hbar\sigma_{xx}}x^{2}} \label{Gauss1}%
\end{equation}
which has Fourier transform%
\begin{equation}
\widehat{\psi}(p)=\left(  \tfrac{1}{2\pi\sigma_{pp}}\right)  ^{1/4}%
e^{-\frac{p^{2}}{4\sigma_{pp}}}e^{-\frac{i\sigma_{xp}}{2\hbar\sigma_{pp}}%
p^{2}} \label{FGauss1}%
\end{equation}
and thus
\begin{equation}
|\psi(x)|^{2}=\left(  \tfrac{1}{2\pi\sigma_{xx}}\right)  ^{1/2}e^{-\frac
{x^{2}}{2\sigma_{xx}}}\text{ \ },\text{ \ }|\widehat{\psi}(p)|^{2}=\left(
\tfrac{1}{2\pi\sigma_{xx}}\right)  ^{1/2}e^{-\frac{p^{2}}{2\sigma_{pp}}}
\label{marg1}%
\end{equation}
Since we have
\begin{equation}
\sigma_{xx}\sigma_{pp}-\sigma_{xp}^{2}=\tfrac{1}{4}\hbar^{2} \label{RSUP1}%
\end{equation}
the covariance $\sigma_{xp}$ can take any of the two values $\pm(\sigma
_{xx}\sigma_{pp}-\tfrac{1}{4}\hbar^{2})^{1/2}$ so the Pauli problem has two
possible solutions $\psi^{\pm}$ (\textquotedblleft Pauli
partners\textquotedblright:\ see Corbett's \cite{Corbett} review of Pauli's
problem). The same argument works for multivariate Gaussians (\ref{psixy}):
setting$X=\frac{\hbar}{2}\Sigma_{XX}^{-1}$ and $Y=-\frac{2}{\hbar}\Sigma
_{XP}\Sigma_{XX}^{-1}$ we have%
\begin{equation}
\psi_{X,Y}(x)=\left(  \tfrac{1}{2\pi}\right)  ^{n/4}(\det\Sigma_{XX}%
)^{-1/4}\exp\left[  -\left(  \frac{1}{4}\Sigma_{XX}^{-1}+\frac{i}{2\hbar
}\Sigma_{XP}\Sigma_{XX}^{-1}\right)  x\cdot x\right]  \label{Gauss3}%
\end{equation}
from which one infers that
\begin{align}
|\psi_{X,Y}(x)|^{2}  &  =\left(  \tfrac{1}{2\pi}\right)  ^{n/2}(\det
\Sigma_{XX})^{-1/2}\exp\left(  -\frac{1}{2}\Sigma_{XX}^{-1}x\cdot x\right) \\
|\widehat{\psi}_{X,Y}(p)|^{2}  &  =\left(  \tfrac{1}{2\pi}\right)  ^{n/2}%
(\det\Sigma_{PP})^{-1/2}\exp\left(  -\frac{1}{2}\Sigma_{PP}^{-1}p\cdot
p\right)  .
\end{align}
To find the covariance matrix $\Sigma_{XP}$ one then uses the
Robertson--Schr\"{o}dinger formula (\ref{RS1})%
\[
\Sigma_{XX}\Sigma_{PP}-(\Sigma_{XP})^{2}=\frac{\hbar^{2}}{4}I_{n\times n}%
\]
which has multiple solutions in $\Sigma_{XP}$.

Another way of seeing Pauli's problem is to use the Wigner formalism; recall
that the Wigner transform of $\psi\in L^{2}(\mathbb{R}^{n})$ is defined by
\begin{equation}
W\psi(x,p)=\left(  \frac{1}{2\pi\hbar}\right)  ^{n}\int_{\mathbb{R}^{n}%
}e^{-\frac{i}{\hbar}p\cdot y}\psi(x+\tfrac{1}{2}y)\psi^{\ast}(x-\tfrac{1}%
{2}y)dy \label{Wigner}%
\end{equation}
and if $\psi,\widehat{\psi}\in L^{1}(\mathbb{R}^{n})\cap L^{2}(\mathbb{R}%
^{n})$ then the marginal properties%
\begin{equation}
\int_{\mathbb{R}^{n}}W\psi(x,p)dp=|\psi(x)|^{2}\text{ },\text{ }%
\int_{\mathbb{R}^{n}}W\psi(x,p)dx=|\widehat{\psi}(p)|^{2} \label{wigmarg}%
\end{equation}
hold \cite{WIGNER}. Since $W\psi^{\ast}(x,p)=W\psi(x,-p)$ the functions $\psi$
and $\psi^{\ast}$ have the same marginals, which is reflected by the relations
(\ref{marg1}) which are satisfied by both $\psi$ and its complex conjugate.
This idea extends to more general phase retrieval problems, see \cite{Grohs}.

\subsection{Lagrangian frames}

We will henceforth call \textit{Lagrangian frame} a pair $(\ell,\ell^{\prime
})$ of transverse Lagrangian subspaces (see APPENDIX\ B), that is $(\ell
,\ell^{\prime})\in\operatorname*{Lag}(n)\times\operatorname*{Lag}(n)$ and
$\ell\cap\ell^{\prime}=0$ (equivalently $\ell\oplus\ell^{\prime}%
=\mathbb{R}^{2n}$ since $\dim\ell=\dim\ell^{\prime}=n$). When $\ell$ and
$\ell^{\prime}$ are coordinate Lagrangian planes $\ell_{X}$ and $\ell_{P}$ we
will call it the \textit{canonical Lagrangian frame}. We denote by
$\mathcal{F}_{\operatorname*{Lag}}(n)$ The following property will be
essential for our constructions to come (See \cite{Birk}):

\begin{lemma}
\label{LemmaTrans}The symplectic group $\operatorname*{Sp}(n)$ acts
transitively on $\mathcal{F}_{\operatorname*{Lag}}(n)$. In particular, every
Lagrangian frame $(\ell,\ell^{\prime})\in\mathcal{F}_{\operatorname*{Lag}}(n)$
can be obtained from the canonical frame $(\ell_{X},\ell_{P})$ by a symplectic transformation.
\end{lemma}

\begin{proof}
Let $(\ell_{1},\ell_{1}^{\prime})$ and $(\ell_{2},\ell_{2}^{\prime})$ be two
Lagrangian frames. Choose a basis $(e_{1i})_{1\leq1\leq n}$ of $\ell_{1}$ and
a basis $(f_{1j})_{1\leq j\leq n}$ of $\ell_{1}^{\prime}$ whose union
$(e_{1i})_{1\leq1\leq n}\cup(f_{1j})_{1\leq j\leq n}$ is a symplectic basis,
that is $\omega(e_{1i},e_{1j})=\omega(f_{1i},f_{1j})=0$ and .$\omega
(f_{1i},e_{1j})=\delta_{ij}$ for $1\leq i,j\leq n$. Similarly, choose bases
$(e_{2i})_{1\leq1\leq n}$ and $(f_{2j})_{1\leq j\leq n}$ of $\ell_{2}$ and
$\ell_{2}^{\prime}$ whose union is also a symplectic basis. The linear
automorphism of $\mathbb{R}^{2n}$ defined by $S(e_{1i})=e_{2i}$ and
$S(f_{1i})=f_{2i}$ for $1\leq i\leq n$ is in $\operatorname*{Sp}(n)$ and we
have $(\ell_{2},\ell_{2}^{\prime})=(S\ell_{1},S\ell_{1}^{\prime})$.
\end{proof}

\subsection{Gaussian reconstruction by partial tomography}

We are going to generalize the reconstruction procedure for Gaussians using
the notion of Lagrangian frame introduced above. For this we need to define
the integral of a real integrable function $\rho$ on phase space along an
affine subspace $\ell(z)=\ell+z$ where $\ell$ is a Lagrangian subspace of
$(\mathbb{R}^{2n},\omega)$. Assuming \cite{ICP,Birk} that $\ell$ is
represented by the system of equations $Ax+Bp=0$ with $A^{T}B$ (and $B^{T}A$)
symmetric, and $A^{T}A+B^{T}B=I_{n\times n}$ (and $A^{T}A+B^{T}B=I_{n\times
n}$) we parametrize $\ell(z)$ by $z(u)=(x(u),p(u))$, $u\in\mathbb{R}^{n}$,
with
\begin{equation}
x(u)=-B^{T}u+x\text{ \ },\text{ \ }p(u)=A^{T}u+p \label{parameter1}%
\end{equation}
where $z=(x,p)\in\mathbb{R}^{2n}$ is arbitrary. Taking into account the fact
that $A^{T}A+B^{T}B=I_{n\times n}$ we then define the generalized line
integral
\begin{equation}
\int_{\ell(z)}\rho(s)ds=\int_{\mathbb{R}^{n}}\rho(-B^{T}u+x,A^{T}u+p)du.
\label{integral}%
\end{equation}
Applying this definition to the case where $\rho=W\psi$ with $\psi\in
L^{1}(\mathbb{R}^{n})\cap L^{2}(\mathbb{R}^{n})$ and choosing $\ell=\ell_{P}$
we have $B=0$ and $A=I_{n\times n}$ so that, in view of the marginal
properties (\ref{wigmarg}),%
\[
\int_{\ell_{P}(x,p)}W\psi(s)ds=\int_{\mathbb{R}^{n}}W\psi(x,u+p)du=|\psi
(x)|^{2}%
\]
and, similarly, choosing $\ell=\ell_{P}$, and $A=0$, $B=I_{n\times n}$,%
\[
\int_{\ell_{X}}W\psi(s)ds=\int_{\mathbb{R}^{n}}W\psi(-u+x,p)du=|\widehat{\psi
}(p)|^{2}.
\]
We'll generalize these relations below, but we first prove that for fixed
$z=(x,p)$ the integral (\ref{integral}) is independent of the choice of
parametrization (\ref{parameter1}). To see this, we note that the Lagrangian
subspace $\ell$ is the image of $\ell_{X}=\{(u,0):u\in\mathbb{R}^{n}\}$ by the
symplectic rotation (\cite{ICP,Birk}; Appendix A)%
\[
R=%
\begin{pmatrix}
-B^{T} & -A^{T}\\
A^{T} & -B^{T}%
\end{pmatrix}
\in\operatorname*{Sp}(n)\cap O(2n,\mathbb{R})
\]
since $A^{T}A+B^{T}B=I_{n\times n}$ and $B^{T}A$ symmetric. Let a new
parametrization of $\ell(z)$ be
\[
x^{\prime}(u)=-B^{\prime T}u+x\ ,\ \ p^{\prime}(u)=A^{\prime T}u+p
\]
the matrices $A^{\prime}$ and $B^{\prime}$ satisfying relations similar to
those of $A$ and $B$, and let $R^{\prime}$ be the corresponding symplectic
rotation. The product $R^{-1}R^{\prime}$ leaves $\ell_{X}$ invariant; since it
is a symplectic rotation we have $R^{-1}R^{\prime}=%
\begin{pmatrix}
H & 0\\
0 & H
\end{pmatrix}
$ with $H\in O(n,\mathbb{R})$ hence the re-parametrization is%
\[
x^{\prime}(u)=-B^{T}Hu+x\ ,\ p^{\prime}(u)=A^{T}Hu+p
\]
leading to the same value of the integral \ref{integral}) since $d(Hu)=du$.

The following result relates the integral (\ref{integral}) to the notion of
marginal value of the Wigner transform:

\begin{proposition}
\label{PropRadon}Let $\ell\in\operatorname*{Lag}(n)$ and $z=(x,p)$ be as
above. For $\psi\in L^{1}(\mathbb{R}^{n})\cap L^{2}(\mathbb{R}^{n})$) we have
\begin{equation}
\int_{\ell(z)}W\psi(s)ds=|\widehat{U}\psi(Ax+Bp)|^{2} \label{integralbis}%
\end{equation}
where $\widehat{U}\in\operatorname*{Mp}(n)$ covers the symplectic rotation $U=%
\begin{pmatrix}
A & B\\
-B & A
\end{pmatrix}
$.
\end{proposition}

\begin{proof}
The Lagrangian subspace $\ell$ is the image of the momentum space $\ell_{P}$
by the symplectic rotation
\[
U^{T}=U^{-1}=%
\begin{pmatrix}
A^{T} & -B^{T}\\
B^{T} & A^{T}%
\end{pmatrix}
\]
hence, using the covariance relation $W\psi\circ U^{-1}=W(\widehat{U}\psi)$
\cite{Birk,WIGNER},
\begin{align*}
\int_{\ell(z)}W\psi(s)ds  &  =\int_{\mathbb{R}^{n}}W\psi(U^{-1}%
((0,u)+U(x,p)))du\\
&  =\int_{\mathbb{R}^{n}}W(\widehat{U}\psi)(0,u)+U(x,p))du\\
&  =\int_{\mathbb{R}^{n}}W(\widehat{U}\psi)(Ax+Bp,Bx-Ap+u)du\\
&  =\int_{\mathbb{R}^{n}}W(\widehat{U}\psi)(Ax+Bp,u)du
\end{align*}
hence (\ref{integralbis}) in view of the first marginal property
(\ref{wigmarg}).
\end{proof}

Formula (\ref{integralbis}) is closely related to the definition of the
symplectic Radon transform as given in our paper \cite{entropy}. The result
below shows that, however, we do not need the full power of the theory of
inverse Radon transform to reconstruct Gaussians:

\begin{theorem}
\label{ThmRadon}Let $(\ell,\ell^{\prime})\in\mathcal{F}_{\operatorname*{Lag}%
}(n)$ be a Lagrangian frame. The Wigner transform $W\psi_{X,Y}$ (and hence the
Gaussian $\psi_{X,Y}$ itself, up to a unimodular factor) is uniquely
determined by the knowledge of the integrals
\[
\int_{\ell(z)}W\psi_{X,Y}(s)ds\text{ \ and }\int_{\ell^{\prime}(z)}W\psi
_{X,Y}(s)ds
\]
for all $z\in\mathbb{R}^{2n}$.
\end{theorem}

\begin{proof}
It is similar to that of Proposition \ref{PropRadon} above. In view of Lemma
\ref{LemmaTrans} the symplectic group acts transitively on $\mathcal{F}%
_{\operatorname*{Lag}}(n)$ so we can find $S\in\operatorname*{Sp}(n)$ such tat
$(\ell,\ell^{\prime})=S(\ell_{X},\ell_{P})$. Setting $S=%
\begin{pmatrix}
A & B\\
C & D
\end{pmatrix}
$ let $\widehat{S}\in\operatorname*{Mp}(n)$ be one of the two metaplectic
operators covering $S$. Let $\ell(Sz)=S\ell_{X}+Sz$ and $\ell^{\prime
}(Sz)=S\ell_{P}+Sz$; we have, using the covariance relation $W\psi_{X,Y}\circ
S=W(\widehat{S}^{-1}\psi_{X,Y})$,
\begin{align*}
\int_{\ell^{\prime}(Sz)}W\psi_{X,Y}(s)ds  &  =\int_{\mathbb{R}^{n}}W\psi
_{X,Y}\left[  S((0,u)+(x,p))\right]  du\\
&  =\int_{\mathbb{R}^{n}}W(\widehat{S}^{-1}\psi_{X,Y})(x,u+p)du.
\end{align*}
that is, in view of the first marginal formula (\ref{wigmarg}),
\begin{equation}
\int_{\ell^{\prime}(Sz)}W\psi_{X,Y}(s)ds=|\widehat{S}^{-1}\psi_{X,Y}(x)|^{2}.
\label{spsi1}%
\end{equation}
Similarly, using the second formula (\ref{wigmarg}),%
\begin{equation}
\int_{\ell(Sz)}W\psi_{X,Y}(s)ds=|F\widehat{S}^{-1}\psi_{X,Y}(p)|^{2}.
\label{spsi2}%
\end{equation}
These values allow the determination of $\widehat{S}^{-1}\psi_{X,Y}$ and,
hence, of $\psi_{X,Y}.$
\end{proof}

\subsection{Geometric Interpretation}

We now consider the following situation: performing a large number of
measurements on the coordinates $x_{1},...,x_{k},p_{k+1},...,p_{n}$ (with
$1\leq k<n)$ we identify this cloud of measurements with an ellipsoid
$X_{\ell}$ carried by the Lagrangian subspace $\ell$ with coordinates
$x_{1},...,x_{k},p_{k+1},...,p_{n}$ and centered at the origin. We now ask
whether to this ellipsoid we can in some way associate the covariance
ellipsoid $\Omega_{\Sigma}$ of some pure Gaussian state with covariance matrix
$\Sigma$. In view of the discussion above $\Omega_{\Sigma}$ has to be a
quantum blob, \textit{i.e.} the image of the phase space ball $B^{2n}%
(\sqrt{\hbar})$ by some $S\in\operatorname*{Sp}(n)$. The answer is given by
the following result, which actually holds for any Lagrangian subspace (not
necessarily a coordinate subspace). It introduces a notion of polar duality
between Lagrangian subspaces. (The basics of the notion of polar duality for
convex sets are shortly reviewed in Appendix C.)

\begin{theorem}
\label{ThmA}Let $(\ell,\ell^{\prime})\in\mathcal{F}_{\operatorname*{Lag}}(n)$
be a Lagrangian frame in $(\mathbb{R}^{2n},\omega)$. Let $X_{\ell}$ be a
centered ellipsoid carried by $\ell$ and define the dual ellipsoid $(X_{\ell
})_{\ell^{\prime}}^{\hbar}\subset\ell^{\prime}$ by
\begin{equation}
(X_{\ell})_{\ell^{\prime}}^{\hbar}=\{z^{\prime}\in\ell^{\prime}:\sup
\nolimits_{z\in\ell}\omega(z,z^{\prime})\leq\hbar\}. \label{omo3}%
\end{equation}
(i) The John ellipsoid $\Omega$ of the product $X_{\ell}\times(X_{\ell}%
)_{\ell^{\prime}}^{\hbar}$ is a quantum blob $S(B^{2n}(\sqrt{\hbar}))$; (ii)
To that quantum blob corresponds the Gaussian pure state $\psi_{X,Y}^{\pm}=\pm
i^{\mu}\widehat{S}\phi_{0}$ where $\widehat{S}\in\operatorname*{Mp}(n)$ covers
$S$.
\end{theorem}

\begin{proof}
(i) We recall \cite{Ball} that the John ellipsoid of a convex set is the
(unique) ellipsoid of maximal volume contained in that set. Let us first prove
(i) for the particular case $(\ell,\ell^{\prime})=(\ell_{X},\ell_{P})$ and
$X_{\ell_{X}}=B_{X}^{n}(\sqrt{\hbar})$. In this case $(X_{\ell_{X}})_{\ell
_{P}}^{\hbar}=B_{P}^{n}(\sqrt{\hbar})$. We claim that the John ellipsoid of
$B_{X}^{n}(\sqrt{\hbar})\times B_{P}^{n}(\sqrt{\hbar})$ is the trivial quantum
blob $B^{2n}(\sqrt{\hbar})$. To see this we first note that the inclusion
\begin{equation}
B^{2n}(\sqrt{\hbar})\subset B_{X}^{n}(\sqrt{\hbar})\times B_{P}^{n}%
(\sqrt{\hbar}) \label{incl}%
\end{equation}
is obvious, and that we cannot have
\[
B^{2n}(R)\subset B_{X}^{n}(\sqrt{\hbar})\times B_{P}^{n}(\sqrt{\hbar})
\]
if $R>\sqrt{\hbar}$. Assume now that the John ellipsoid of $B_{X}^{n}%
(\sqrt{\hbar})\times B_{P}^{n}(\sqrt{\hbar})$ is defined by the inequality
\[
Ax\cdot x+Bx\cdot p+Cp\cdot p\leq\hbar
\]
where $A,C>0$ and $B$ are real $n\times n$ matrices. Since $B_{X}^{n}%
(\sqrt{\hbar})\times B_{P}^{n}(\sqrt{\hbar})$ is invariant by the
transformation $(x,p)\longmapsto(p,x)$ so is its John ellipsoid and we must
thus have $A=C$ and $B=B^{T}$. Similarly, $B_{X}^{n}(\sqrt{\hbar})\times
B_{P}^{n}(\sqrt{\hbar})$ being invariant by the partial reflection
$(x,p)\longmapsto(-x,p)$ we get $B=0$ so the John ellipsoid of $B_{X}%
^{n}(\sqrt{\hbar})\times B_{P}^{n}(\sqrt{\hbar})$ is defined by $Ax\cdot
x+Ap\cdot p\leq\hbar$. We next observe that $B_{X}^{n}(\sqrt{\hbar})\times
B_{P}^{n}(\sqrt{\hbar})$ an its John ellipsoid are invariant under the
transformations $(x,p)\longmapsto(Hx,HP)$ where $H\in O(n,\mathbb{R})$ so we
must have $AH=HA$ for all $H\in O(n,\mathbb{R})$, but this is only possible if
$A=\lambda I_{n\times n}$ for some $\lambda\in\mathbb{R}$. The John ellipsoid
of $B_{X}^{n}(\sqrt{\hbar})\times B_{P}^{n}(\sqrt{\hbar})$ is thus of the type
$B^{2n}(\sqrt{\hbar}/\sqrt{\lambda})$ for some $\lambda\geq1$ and this
concludes the proof in view of the inclusion (\ref{incl}) since the case
$\lambda>R^{2}$ is excluded. \ Consider now, again when $(\ell,\ell^{\prime
})=(\ell_{X},\ell_{P})$, the case where $X_{\ell_{X}}=A(B_{X}^{n}(\sqrt{\hbar
}))$ is an ellipsoid ($A\in\operatorname*{Sym}_{++}(n,\mathbb{R})$). We then
have $(X_{\ell_{X}})_{\ell_{P}}^{\hbar}=A^{-1}(B_{P}^{n}(\sqrt{\hbar}))$ so
that
\[
X_{\ell_{X}}\times(X_{\ell_{X}})_{\ell_{P}}^{\hbar}=S_{A}(B_{X}^{n}%
(\sqrt{\hbar})\times B_{P}^{n}(\sqrt{\hbar}))
\]
where $S_{A}=%
\begin{pmatrix}
A & 0\\
0 & A^{-1}%
\end{pmatrix}
\in\operatorname*{Sp}(n)$ so that the John ellipsoid of $X_{\ell_{X}}%
\times(X_{\ell_{X}})_{\ell_{P}}^{\hbar}$ is here the quantum blob
$S_{A}(B^{2n}(\sqrt{\hbar}))$. Let us finally consider the case of an
arbitrary Lagrangian frame $(\ell,\ell^{\prime})$. In view of the transitivity
of the action of $\operatorname*{Sp}(n)$ on $\mathcal{F}_{\operatorname*{Lag}%
}(n)$ we can find $S\in\operatorname*{Sp}(n)$ such that $(\ell,\ell^{\prime
})=S(\ell_{X},\ell_{P})$; set now $X_{\ell_{X}}=S^{-1}(X_{\ell})$, this is an
ellipsoid carried by $(\ell_{X}$. We claim that we have $(X_{\ell_{X}}%
)_{\ell_{P}}^{\hbar}=S^{-1}((X_{\ell})_{\ell^{\prime}}^{\hbar})$. Let $z\in
S^{-1}(X_{\ell})_{\ell^{\prime}}^{\hbar}$, that is $Sz\in(X_{\ell}%
)_{\ell^{\prime}}^{\hbar}$. This is equivalent to the conditions $z\in
S^{-1}\ell^{\prime}=\ell_{P}$ and $\omega(Sz,z^{\prime})\leq\hbar$ for all
$z^{\prime}\in X_{\ell}$. Since $\omega(Sz,z^{\prime})=\omega(z,S^{-1}%
z^{\prime})$ this is in turn equivalent to $z\in S^{-1}\ell^{\prime}=\ell_{P}$
and $\omega(z,S^{-1}z^{\prime})\leq\hbar$ for all $S^{-1}z^{\prime}\in
S^{-1}(X_{\ell})$, that is to $z\in\omega(z,z^{\prime\prime})\leq\hbar$ for
all $z^{\prime\prime}\in X_{\ell_{X}}$ which is the same thing as
$z\in(X_{\ell})_{\ell^{\prime}}^{\hbar}$ which was to be proven. Summarizing,
we have shown that
\[
X_{\ell}\times(X_{\ell})_{\ell^{\prime}}^{\hbar}=S(X_{\ell_{X}}\times
(X_{\ell_{X}})_{\ell_{P}}^{\hbar}).
\]
Since the John ellipsoid of $X_{\ell_{X}}\times(X_{\ell_{X}})_{\ell_{P}%
}^{\hbar}$ is $S_{A}(B^{2n}(\sqrt{\hbar}))$, that of $X_{\ell}\times(X_{\ell
})_{\ell^{\prime}}^{\hbar}$ is the quantum blob $Q=SS_{A}(B^{2n}(\sqrt{\hbar
}))$. (ii) (\textit{cf}. \cite{blob}). The result follows from the discussion
above. Set $S^{\prime}=SS_{A}$. To $Q=S^{\prime}(B^{2n}(\sqrt{\hbar}))$ we
associate the positive definite symplectic matrix $G=S^{\prime}S^{\prime T}.$
To the latter corresponds the phase space Gaussian $z\longmapsto(\pi
\hbar)^{-n}e^{-\frac{1}{\hbar}Gz\cdot z}$, which is the Wigner transform of
$i^{\mu}\widehat{S^{\prime}}\phi_{0}$ for any $\mu\in\mathbb{R}$.
\end{proof}

\section{Polar Duality and Covariance Ellipsoid}

\subsection{Symplectic polar duality}

Here is another variant of polar duality; it was introduced in our paper
\cite{Bull} and used in \cite{sym} to characterize covariance ellipsoids (we
are following quite closely the presentation in the latter paper). Let
$\Omega$ be a symmetric convex body in the phase space $(\mathbb{R}%
^{2n},\omega)$ (for instance an ellipsoid). We call \textit{symplectic
}($\hbar$-)\textit{polar dual} $\Omega^{\hbar,\omega}$ of $\Omega$ the set
\begin{equation}
\Omega^{\hbar,\omega}=\{z^{\prime}\in\mathbb{R}^{2n}:\sup\nolimits_{z\in
\Omega}\omega(z,z^{\prime})\leq\hbar\}. \label{omegapol1bis}%
\end{equation}
It is related to the usual polar dual
\[
\Omega^{\hbar}=\{z^{\prime}\in\mathbb{R}^{2n}:\sup\nolimits_{z\in\Omega
}(z\cdot z^{\prime})\leq\hbar\}
\]
by a symplectic rotation:
\begin{equation}
\Omega^{\hbar,\omega}=(J\Omega)^{\hbar}=J(\Omega^{\hbar}). \label{omegapol2}%
\end{equation}
The symplectic polar dual is symplectically covariant in the sense that
\begin{equation}
(S\Omega)^{\hbar,\omega}=S(\Omega^{\hbar,\omega})\text{ \textit{for} }%
S\in\operatorname*{Sp}(n). \label{omegapol3}%
\end{equation}
In fact, the condition $S\in\operatorname*{Sp}(n)$ is equivalent to
$S^{T}JS=J$ hence $JS=(S^{T})^{-1}J$. It follows that
\begin{align*}
(S\Omega)^{\hbar,\omega}  &  =J(S(\Omega))^{\hbar}=J(S^{T})^{-1}(\Omega
^{\hbar})\\
&  =SJ(\Omega^{\hbar})=S(\Omega^{\hbar,\omega})
\end{align*}
which is (\ref{omegapol3}). In particular, since $B^{2n}(\sqrt{\hbar})^{\hbar
}=B^{2n}(\sqrt{\hbar})$, we have
\begin{equation}
(S(B^{2n}(\sqrt{\hbar})))^{\hbar,\omega}=S(B^{2n}(\sqrt{\hbar}))
\label{fixblob}%
\end{equation}
so\textit{ }the quantum blobs $S(B^{2n}(\sqrt{\hbar}))$, $S\in
\operatorname*{Sp}(n)$, are fixed points of the transformation $\Omega
\longmapsto\Omega^{\hbar,\omega}$; it is easy to show that they are the only
fixed points of this transformation.

\begin{proposition}
\label{Prop3}Let $\Omega_{\Sigma}$ be the covariance ellipsoid associated with
the covariance matrix $\Sigma$:
\begin{equation}
\Omega_{\Sigma}=\{z\in\mathbb{R}^{2n}:Mz\cdot z\leq\hbar\}\text{ \ },\text{
\ }M=\tfrac{1}{2}\hbar\Sigma^{-1}.
\end{equation}
(i) $\Omega_{\Sigma}$ is quantized (that is, contains a quantum blob) if and
only if we have the inclusion $\Omega_{\Sigma}^{\hbar,\omega}\subset
\Omega_{\Sigma}$. (ii) The equality $\Omega_{\Sigma}^{\hbar,\omega}%
=\Omega_{\Sigma}$ holds if and only if there exists $S\in\operatorname*{Sp}%
(n)$ such that $\Omega_{\Sigma}=S(B^{2n}(\sqrt{\hbar}))$ (i.e. if and only if
$\Omega_{\Sigma}$ is a quantum blob).
\end{proposition}

\begin{proof}
(i) Suppose that there exists $S\in\operatorname*{Sp}(n)$ such that
$Q=S(B^{2n}(\sqrt{\hbar}))\subset\Omega$. By the anti-monotonicity of
(symplectic) polar duality this implies that we have $\Omega^{\hbar,\omega
}\subset Q^{\hbar,\omega}=Q\subset\Omega$, which proves the necessity of the
condition. Suppose conversely that we have $\Omega^{\hbar,\omega}\subset
\Omega$. Since%
\begin{equation}
\Omega_{\Sigma}^{\hbar}=\{z\in\mathbb{R}^{2n}:M^{-1}z\cdot z\leq\hbar\}
\end{equation}
we have, using (\ref{omegapol2}),
\begin{equation}
\Omega_{\Sigma}^{\hbar,\omega}=\{z\in\mathbb{R}^{2n}:(-JM^{-1}J)z\cdot
z\leq\hbar\} \label{ommom}%
\end{equation}
hence the inclusion $\Omega_{\Sigma}^{\hbar,\omega}\subset\Omega_{\Sigma}$ is
equivalent to $M\leq(-JM^{-1}J)$ ($\leq$ stands here for the L\"{o}wner
ordering). Performing a symplectic diagonalization (\ref{Williamson}) of $M$
and using the relations $JS^{-1}=S^{T}J$, $(S^{T})^{-1}J=JS$ this is
equivalent to
\[
M=S^{T}DS\leq S^{T}(-JD^{-1}J)S
\]
that is to $D\leq-JD^{-1}J$. In the notation in (\ref{Williamson}) this
implies that we have $\Lambda^{\omega}\leq(\Lambda^{\omega})^{-1}$ and hence
$\lambda_{j}^{\omega}\leq1$ for $1\leq j\leq n$; thus $D\leq I$ and
$M=S^{T}DS\leq S^{T}S$. The inclusion $S(B^{2n}(\sqrt{\hbar}))\subset\Omega$
follows. Let us next prove the statement (ii) The condition is sufficient
since $S(B^{2n}(\sqrt{\hbar}))^{\hbar,\omega}=S(B^{2n}(\sqrt{\hbar}))$. Assume
conversely that $\Omega_{\Sigma}^{\hbar,\omega}=\Omega_{\Sigma}$. Then there
exists $S\in\operatorname*{Sp}(n)$ such that$Q=S(B^{2n}(\sqrt{\hbar}%
))\subset\Omega_{\Sigma}$. It follows that $\Omega_{\Sigma}^{\hbar,\omega
}\subset Q^{\hbar,\omega}=Q$ hence $\Omega_{\Sigma}^{\hbar,\omega}%
=\Omega_{\Sigma}\subset Q$ so that $\Omega=Q$.
\end{proof}

Here is an example:

\begin{example}
\label{example2}Consider again the covariance ellipse
\begin{equation}
\Omega_{\Sigma}:\dfrac{\sigma_{pp}}{2D}x^{2}-\frac{\sigma_{xp}}{D}%
px+\dfrac{\sigma_{xx}}{2D}p^{2}\leq1
\end{equation}
with $D=\sigma_{xx}\sigma_{pp}-\sigma_{xp}^{2}=\tfrac{1}{4}\hbar^{2}$. Here
$M=\frac{\hbar}{2D}%
\begin{pmatrix}
\sigma_{pp} & -\sigma_{xp}\\
-\sigma_{xp} & \sigma_{xx}%
\end{pmatrix}
$ hence the symplectic polar dual of $\Omega_{\Sigma}$ is%
\[
\Omega_{\Sigma}^{\hbar,\omega}:\frac{2\sigma_{pp}}{\hbar^{2}}x^{2}%
-\frac{4\sigma_{xp}}{\hbar^{2}}px+\frac{2\sigma_{xx}}{\hbar^{2}}p^{2}\leq1.
\]
The condition $\Omega_{\Sigma}^{\hbar,\omega}=\Omega_{\Sigma}$ is equivalent
to $D=\frac{1}{4}\hbar^{2}$ so that $\Omega_{\Sigma}$ is indeed a quantum blob.
\end{example}

\subsection{A tomographic result}

We are going to prove a stronger statement, which can be seen as a
\textquotedblleft tomographic\textquotedblright\ result since it involves the
intersection of the covariance ellipsoid with a (Lagrangian) subspace. Let us
begin with a simple example in the case $n=1$.

\begin{theorem}
\label{ThmB}(i) The ellipsoid $\Omega$ contains a quantum blob $Q=S(B^{2n}%
(\sqrt{\hbar}))$ ($S\in\operatorname*{Sp}(n)$) if and only if there exists
$\ell\in\operatorname*{Lag}(n)$ such that%
\begin{equation}
\Omega^{\hbar,\omega}\cap\ell\subset\Omega\cap\ell\label{interl}%
\end{equation}
in which case we have $\Omega^{\hbar,\omega}\cap\ell\subset\Omega\cap\ell$ for
all $\ell\in\operatorname*{Lag}(n)$. (ii) The equality $\Omega^{\hbar,\omega
}\cap\ell=\Omega\cap\ell$ holds \ (for some, and hence for all, $\ell
\in\operatorname*{Lag}(n)$) if and only if $\Omega$ is a quantum blob.
\end{theorem}

\begin{proof}
(i) The necessity of the condition (\ref{interl}) is trivial (Proposition
\ref{Prop3}). Let us prove that this condition is also sufficient. Setting as
usual $M=\frac{1}{2}\hbar\Sigma^{-1}$ and
\begin{equation}
\Omega_{\Sigma}=\{z:Mz\cdot z\leq\hbar\}=\{z:\tfrac{1}{2}\Sigma^{-1}z\cdot
z\leq\hbar\}
\end{equation}
we have%
\begin{equation}
\Omega_{\Sigma}^{\hbar,\omega}=\{z\in\mathbb{R}^{2n}:(-JM^{-1}J)z\cdot
z\leq\hbar\}. \label{omegasdual}%
\end{equation}
Performing a symplectic diagonalization (\ref{Williamson}) of $M$ we get%
\begin{equation}
\Omega_{\Sigma}=S^{-1}\Omega_{\hbar D^{-1}/2}\text{ },\text{ \ }\Omega
_{\Sigma}^{\hbar,\omega}=S^{-1}(\Omega_{\hbar D^{-1}/2})^{\hbar,\omega}
\label{oms}%
\end{equation}
where $\Omega_{\hbar D^{-1}/2}$ and its dual are explicitly given by
\begin{align*}
\Omega_{\hbar D^{-1}/2}  &  =\{z\in\mathbb{R}^{2n}:Dz\cdot z\leq\hbar\}\\
(\Omega_{\hbar D^{-1}/2})^{\hbar,\omega}  &  =\{z\in\mathbb{R}^{2n}%
:D^{-1}z\cdot z\leq\hbar\}.
\end{align*}
where we have used the identity $-JD^{-1}J=D^{-1}$. Let us first assume that
$\ell=\ell_{X}=\mathbb{R}^{n}\times0$. Then
\[
\Omega_{\hbar D^{-1}/2}\cap\ell_{X}=\{x\in\mathbb{R}^{n}:\Lambda^{\omega
}x\cdot x\leq\hbar\}
\]
and
\[
(\Omega_{\hbar D^{-1}/2})^{\hbar,\omega}\cap\ell_{X}=\{x\in\mathbb{R}%
^{n}:(\Lambda^{\omega})^{-1}x\cdot x\leq\hbar\}.
\]
Now, the condition
\[
(\Omega_{\hbar D^{-1}/2})^{\hbar,\omega}\cap\ell_{X}\subset\Omega_{\hbar
D^{-1}/2}\cap\ell_{X}%
\]
is equivalent to $(\Lambda^{\omega})^{-1}\geq\Lambda^{\omega}$ that is to
$D^{-1}\geq D$, which implies $(\Omega_{\hbar D^{-1}/2})^{\hbar,\omega}%
\subset\Omega_{\hbar D^{-1}/2}$, and $\Omega_{\hbar D^{-1}/2}$ contains a
quantum blob in view of Proposition \ref{Prop3}.. We have thus proven our
result in the case where $\Sigma=\hbar D^{-1}/2$ and $\ell=\ell_{X}$. For the
general case we take $\ell=S^{-1}\ell_{X}$ where $S$ is a Williamson
diagonalizing matrix for $\Sigma$; in view of (\ref{oms}) we have
\begin{align*}
\Omega_{\Sigma}\cap\ell &  =S^{-1}\Omega_{\hbar D^{-1}/2}\cap S^{-1}\ell
_{X}=S^{-1}(\Omega_{\hbar D^{-1}/2}\cap\ell_{X})\\
\Omega_{\Sigma}^{\hbar,\omega}\cap\ell &  =S^{-1}(\Omega_{\hbar D^{-1}%
/2})^{\hbar,\omega}\cap S^{-1}\ell_{X}=S^{-1}((\Omega_{\hbar D^{-1}/2}%
)^{\hbar,\omega}\cap\ell_{X})
\end{align*}
and hence $\Omega_{\Sigma}^{\hbar,\omega}\cap\ell\subset\Omega_{\Sigma}%
\cap\ell$ if and only if $(\Omega_{\hbar D^{-1}/2})^{\hbar,\omega}%
\subset\Omega_{\hbar D^{-1}/2}$. It now suffices to apply Proposition
\ref{Prop3}. To prove (ii) it is sufficient to note that the equality
\[
(\Omega_{\hbar D^{-1}/2})^{\hbar,\omega}\cap\ell_{X}=\Omega_{\hbar D^{-1}%
/2}\cap\ell_{X}%
\]
is equivalent to $(\Lambda^{\omega})^{-1}=\Lambda^{\omega}$ that is to
$\Lambda^{\omega}=I_{n\times n}$. Since we have in this case $M=S_{0}^{T}%
S_{0}$ in view of (\ref{Williamson}), the proof in the general case can be
completed as above.
\end{proof}

\begin{example}
\label{example3}With the notation of the previous examples, we have
\begin{align}
\Omega_{\Sigma}  &  :\dfrac{\sigma_{pp}}{2D}x^{2}-\frac{\sigma_{xp}}%
{D}px+\dfrac{\sigma_{xx}}{2D}p^{2}\leq1\\
\Omega_{\Sigma}^{\hbar,\omega}  &  :\frac{2\sigma_{pp}}{\hbar^{2}}x^{2}%
-\frac{4\sigma_{xp}}{\hbar^{2}}px+\frac{2\sigma_{xx}}{\hbar^{2}}p^{2}\leq1
\end{align}
with $D=\sigma_{xx}\sigma_{pp}-\sigma_{xp}^{2}=\tfrac{1}{4}\hbar^{2}$. We
have
\begin{align*}
\Omega_{\Sigma}\cap\ell_{X}  &  =[-\sqrt{2D/\sigma_{pp}},\sqrt{2D/\sigma_{pp}%
}]\\
\Omega_{\Sigma}^{\hbar,\omega}\cap\ell_{X}  &  =[-\hbar/\sqrt{2\sigma_{pp}%
},\hbar/\sqrt{2\sigma_{pp}}]
\end{align*}
and $\Omega_{\Sigma}^{\hbar,\omega}\cap\ell_{X}\subset\Omega_{\Sigma}\cap
\ell_{X}$ if and only if $\hbar/\sqrt{2\sigma_{pp}}\leq\sqrt{2D/\sigma_{pp}}$
which is equivalent to $D\geq\frac{1}{4}\hbar^{2}$. More generally, if
$\ell_{a}:p=ax$ for any $a\in\mathbb{R}$ then $\Omega_{\Sigma}\cap\ell_{a}$
and $\Omega_{\Sigma}^{\hbar,\omega}\cap\ell_{a}$ are determined by the
inequalities%
\begin{gather*}
\Omega_{\Sigma}\cap\ell_{a}:\left(  \dfrac{\sigma_{pp}}{2D}-\frac{\sigma_{xp}%
}{D}a+\dfrac{\sigma_{xx}}{2D}a^{2}\right)  x^{2}\leq1\\
\Omega_{\Sigma}^{\hbar,\omega}\cap\ell_{a}:\left(  \frac{2\sigma_{pp}}%
{\hbar^{2}}-\frac{4\sigma_{xp}}{\hbar^{2}}a+\frac{2\sigma_{xx}}{\hbar^{2}%
}a^{2}\right)  x^{2}\leq1
\end{gather*}
and the inclusion
\begin{equation}
\Omega_{\Sigma}^{\hbar,\omega}\cap\ell_{a}\subset\Omega_{\Sigma}\cap\ell_{a}
\label{inclusa}%
\end{equation}
is equivalent to the inequality
\begin{equation}
k\left(  \sigma_{xx}a^{2}-2\sigma_{xp}a+\sigma_{pp}\right)  \leq0 \label{k}%
\end{equation}
where $k=2(\hbar^{2}/4-D)/\hbar^{2}D$. Now $\sigma_{xx}a^{2}-2\sigma
_{xp}a+\sigma_{pp}>0$ for every $a\in\mathbb{R}$ (because $\sigma_{xp}%
^{2}-\sigma_{xx}\sigma_{pp}=-D<0$ since $\Sigma>0$) and hence the inclusion
(\ref{inclusa}) holds if and only if $k\leq0$, that is, if and only if
$D\geq\hbar^{2}/4$ which is the Robertson--Schr\"{o}dinger inequality ensuring
us that $\Omega_{\Sigma}$ contains a quantum blob (and is itself a quantum
blob when $D=\hbar^{2}/4$).
\end{example}

\subsection{The case of mixed states\label{secmixed}}

Sofar we have been considering pure states. Let us generalize our discussion
to more general mixed states. We assume that $\widehat{\rho}$ is what we have
called in \cite{QHA} a \textquotedblleft Feichtinger state\textquotedblright,
\textit{i.e}. a density operator whose Wigner distribution $\rho$ is regular
enough to allow the existence of the covariance matrix%
\begin{equation}
\Sigma=\int_{\mathbb{R}^{2n}}(z-\langle z\rangle)(z-\langle z\rangle)^{T}%
\rho(z)dz\label{sigmaz1}%
\end{equation}
( $\langle z\rangle=\int_{\mathbb{R}^{2n}}z\rho(z)dz$ is the mean value
vector). In order to represent a quantum state a necessary condition is that
\cite{dutta,Narco}
\begin{equation}
\Sigma+\frac{i\hbar}{2}J\text{ \textit{is semidefinite positive}%
}\label{semidepo}%
\end{equation}
which we write for short as $\Sigma+\frac{i\hbar}{2}J$ $\geq0$ (this condition
equivalent to the uncertainty principle in its strong
Robertson--Schr\"{o}dinger form, \textit{ibid.}). One shows \ \cite{Narco}
that (\ref{semidepo}) implies that the covariance matrix $\Sigma$ of a quantum
state is always definite positive, and, conversely, that (\ref{semidepo}) is
sufficient for Gaussian mixed states: a Gaussian of the type (\ref{rhoG})
introduced above, tha is
\begin{equation}
\rho(z)=\frac{1}{(2\pi)^{n}\sqrt{\det\Sigma}}e^{-\frac{1}{2}\Sigma
^{-1}(z-z_{0})\cdot(z-z_{0})}\label{rhoGbis}%
\end{equation}
is the Wigner distribution of a mixed quantum state if and only if the
condition (\ref{semidepo}) holds. Recall that this contition is equivalent to
saying that the covariance ellipsoid $\Omega_{\Sigma}:\frac{1}{2}\Sigma
^{-1}z\cdot z\leq1$ contains a quantum blob, from which follows that the
orthogonal projetions of $\Omega_{\Sigma}:$ on any conjugate plane
$x_{j},p_{j}$ has area at least $\pi\hbar$.

In view of Theoren \ref{ThmB} $\Omega_{\Sigma}$ contains a quantum blob
$Q=S(B^{2n}(\sqrt{\hbar}))$if and only if there exists $\ell\in
\operatorname*{Lag}(n)$ such that $\Omega^{\hbar,\omega}\cap\ell\subset
\Omega\cap\ell$ in which case we have $\Omega^{\hbar,\omega}\cap\ell
\subset\Omega\cap\ell$ for all $\ell\in\operatorname*{Lag}(n)$. It follows that:

\begin{corollary}
The Gaussian \ref{rhoGbis} is the Wigner distribution of a mixed quantum state
if and only its covariance ellipsoid satifies the condition $\Omega_{\Sigma
}^{\hbar,\omega}\cap\ell\subset\Omega_{\Sigma}\cap\ell$ for some (and hence
all) $\ell\in\operatorname*{Lag}(n)$.
\end{corollary}

\begin{proof}
It is just a restatemnt of Theoren \ref{ThmB}.
\end{proof}

\section{Discussion and Conclusions}

Theorem \ref{ThmA} shows that we can identify any pure Gaussian state with a
geometric object, the Cartesian product $X_{\ell}\times(X_{\ell}%
)_{\ell^{\prime}}^{\hbar}$. The physical interpretation of this correspondence
is the following: once a cloud of position-momentum measurements is made on a
given Lagrangian plane $\ell$, the latter is approximated by an ellipsoid
$X_{\ell}$. One then chooses a transversal Lagrangian plane $\ell^{\prime}$
and one calculates the polar dual $(X_{\ell})_{\ell^{\prime}}^{\hbar}$ of
$X_{\ell}$ on $\ell^{\prime}$; the covariance ellipsoid of the Gaussian state
we are looking for is then simply the maximal volume ellipsoid of the convex
product $X_{\ell}\times(X_{\ell})_{\ell^{\prime}}^{\hbar}$, and the latter
uniquely determines the state (which is here supposed to be centered at the
origin; the general case is trivially obtained using phase space translation
or the Heisenberg displacement operator). Theorem \ref{ThmB}, on the other
hand, shows that one can test whether a covariance ellipsoid is that of a
quantum state by intersecting it with a single arbitrary Lagrangian plane.
This is typically a tomographic result which might have both theoretical and
practical applications.

It would be interesting (and important!) to extend the approach and results of
this paper to non-Gaussian states; non-Gaussian features are indispensable in
many quantum protocols, especially to reach a quantum computational advantage
(see the interesting discussions in Ra \textit{et al}. \cite{Ra} and
Walschaers \cite{Wal}). A possible approach could be to generalize the
\textquotedblleft geometric states\textquotedblright\ described by Theorem
\ref{ThmA} to the case where $X_{\ell}$ no longer is an ellipsoid, but an
arbitrary convex body. The Lagrangian plane $\ell$ would then be replaced with
a Lagrangian submanifold of phase space (i.e. a $n$-dimensional submanifold
where the tangent spaces are all Lagrangian). We will come back to these
intriguing and potentially fruitful possibilities in future work.

An interesting point raised by one of the Reviewers is the question of what
happens for reduced covariance matrices where obtaining purity or von entropy
Neumann is possible? These questions will be answered in a forthcoming article
(There are some delicate points to elucidate, and we have found they deserve
to be discussed in a sequel of this work).

\section*{APPENDIX\ A. The Groups $\operatorname*{Sp}(n)$, $U(n)$, and
$\operatorname*{Mp}(n)$}

For details and proofs see \cite{Birk}. The symplectic group
$\operatorname*{Sp}(n)$ consists of all linear automorphisms $S:\mathbb{R}%
^{2n}\longrightarrow\mathbb{R}^{2n}$ such that $\omega(Sz,Sz^{\prime}%
)=\omega(z,z^{\prime})$ for all $(z,z^{\prime})\in\mathbb{R}^{2n}%
\times\mathbb{R}^{2n}$. A symplectic basis of $(\mathbb{R}_{z}^{2n},\sigma)$
being chosen once for all, we can write this condition in matrix form
$S^{T}JS=SJS^{T}=J$ and, writing $S\in$ $\operatorname*{Sp}(n)$ in
block-matrix form
\begin{equation}
S=%
\begin{pmatrix}
A & B\\
C & D
\end{pmatrix}
\label{sabcd}%
\end{equation}
where the entries $A,B,C,D$ are $n\times n$ matrices, these conditions are
then easily seen to be equivalent to the two following sets of equivalent
conditions:
\begin{align}
A^{T}C\text{, }B^{T}D\text{ \ symmetric, and }A^{T}D-C^{T}B  &
=I\label{cond1}\\
AB^{T}\text{, }CD^{T}\text{ \ symmetric, and }AD^{T}-BC^{T}  &  =I\text{.}
\label{cond2}%
\end{align}
It follows from the second of these sets of conditions that the inverse of $S$
is
\begin{equation}
S^{-1}=%
\begin{pmatrix}
D^{T} & -B^{T}\\
-C^{T} & A^{T}%
\end{pmatrix}
. \label{cond3}%
\end{equation}
There are several ways to describe the generators of the group
$\operatorname*{Sp}(n)$. We will use here the following:
\begin{equation}
V_{-P}=%
\begin{pmatrix}
I & 0\\
-P & I
\end{pmatrix}
\text{ },\text{ }M_{L}=%
\begin{pmatrix}
L^{-1} & 0\\
0 & L^{T}%
\end{pmatrix}
\text{ \ },\text{ \ }J=%
\begin{pmatrix}
0 & I\\
-I & 0
\end{pmatrix}
\label{vpmlj}%
\end{equation}
where $P=P^{T}$ and $\det L\neq0$.

The subgroup $\operatorname*{Sp}(n)\cap O(2n,\mathbb{R})$ of symplectic
rotations is denoted by $U(n)$; this notation comes from the fact that $U(n)$
is identified with the unitary group $U(n,\mathbb{C})$ via the monomorphism
$A+iB\hookrightarrow%
\begin{pmatrix}
A & B\\
-B & A
\end{pmatrix}
$.

The symplectic group is a connected Lie group contractible to $U(n)\equiv
U(n,cC)$ and therefore has covering groups $\operatorname*{Sp}_{q}(n)$ of all
orders; the double covering $\operatorname*{Sp}_{2}(n)$ has a unitary
representation in $L^{2}(\mathbb{R}^{n})$, the metaplectic group
$\operatorname*{Mp}(n)$. The latter is generated by the operators
$\widehat{V}_{-P}$, $\widehat{M}_{L,m}$, and $\widehat{J}$ defined by
\[
\widehat{V}_{-P}\psi(x)=e^{\frac{i}{2\hbar}Px\cdot x}\psi(x)\text{ \ ,
}\widehat{M}_{L,m}\psi(x)=i^{m}\sqrt{|\det L|}\psi(Lx)\text{\ }%
\]
(the integer $m$ corresponding to a choice of $\arg\det L$), and $\widehat{J}$
being essentially the Fourier transform:%

\[
\widehat{J}\psi(x)=\left(  \frac{1}{2\pi i\hbar}\right)  ^{n/2}\int%
_{\mathbb{R}^{n}}e^{-\frac{i}{\hbar}x\cdot x^{\prime}}\psi(x^{\prime
})dx^{\prime}.
\]
Denoting by $\pi^{\operatorname*{Mp}}$ the covering projection
$\operatorname*{Mp}(n)\longrightarrow\operatorname*{Sp}(n)$ we have%
\[
\pi^{\operatorname*{Mp}}(\widehat{V}_{-P})=V_{-P}\text{ },\text{ }%
\pi^{\operatorname*{Mp}}(\widehat{M}_{L,m})=M_{L}\text{ \ },\text{ }%
\pi^{\operatorname*{Mp}}(\widehat{J})=J.\text{\ }%
\]

We will also use the following factorization results: given $S\in
\operatorname*{Sp}(n)$ written in block form (\ref{sabcd}) we have the
pre-Iwasawa factorization \cite{dutta,iwa}: There exist unique matrices
$P=P^{T}$ and $L=L^{T}>0$ and $R\in\operatorname*{Sp}(n)\cap O(2n)$ such that
\begin{equation}
S=V_{P}M_{L}R. \label{iwa1}%
\end{equation}
These matrices are given by%
\begin{gather}
P=-(CA^{T}+DB^{T})(AA^{T}+BB^{T})^{-1}\label{pl1}\\
L=(AA^{T}+BB^{T})^{-1/2}; \label{pl2}%
\end{gather}
writing $R=%
\begin{pmatrix}
E & F\\
-F & E
\end{pmatrix}
$ the $n\times n$ blocks $E$ and $F$ are given by%
\begin{equation}
E=(AA^{T}+BB^{T})^{-1/2}A\text{ \ },\text{ \ }F=(AA^{T}+BB^{T})^{-1/2}B.
\label{unixy}%
\end{equation}
The matrix $R$ is a symplectic rotation: $R\in\operatorname*{Sp}(n)\cap
O(2n,\mathbb{R})$.

\section*{APPENDIX\ B. Lagrangian Subspaces}

By definition a \textit{Lagrangian coordinate subspace} is a $n$-dimensional
subspace $\ell_{(\alpha,\beta)}$ of the $(\mathbb{R}^{2n},\omega)$ given by
the relations $x_{(\alpha)}=0$ and $p_{(\beta)}$ where $\alpha$ and $\beta$
form a partition of the set of integers $\{1,...,n\}$. Thus, for instance,
$x_{1}=0$ and $p_{1}=0$ defines coordinate Lagrangian subspaces in
$\mathbb{R}^{4}$. The choices $\alpha=\emptyset$ and $\beta=\emptyset$
correspond to the canonical coordinate planes $\ell_{X}$ and $\ell_{P}$,
respectively. Let $\ell_{(\alpha,\beta)}$ be a Lagrangian coordinate subspace;
we assume for notational simplicity that $\alpha=\{1,...,k\}$, $\beta
=\{k+1,...,n\}$ ($1\leq k<n$). It is represented by the equation%
\begin{equation}
Ax+Bp=0 \label{AxBp}%
\end{equation}
where $A$ and $B$ are the diagonal matrices $A=I_{k\times k}\otimes
0_{(n-k)\times(n-k)}$ and $B=0_{k\times k}\otimes I_{(n-k)\times(n-k)}$. A
remarkable feature of coordinate Lagrangian subspaces is that the symplectic
form $\omega$ vanishes identically on them if $z,z^{\prime}\in\ell
_{(\alpha,\beta)}$ then $\omega(z,z)=0$. This motivates the following
definition \cite{ICP,Birk}: a $n$-dimensional subspace $\ell$ of
$\mathbb{R}^{2n}$ is called a Lagrangian subspace (or plane) if $\omega
(z,z)=0$ for all $z,z^{\prime}\in\ell$. (Lagrangian subspaces intervene in
many areas of mathematical physics; for instance they are the tangent spaces
to the invariant tori of classical mechanics \cite{Arnold,ICP}). In the case
$n=1$ Lagrangian planes are just the straight lines through the origin in the
phase plane. In the general case they are represented by equations $Ax+Bp=0$
where $\operatorname*{rank}(A,B)=n$ and $A^{T}B=B^{T}A$ \cite{ICP}. It turns
out that every Lagrangian subspace can be obtained from any Lagrangian
coordinate subspace using a symplectic transformation. This follows from the
fact the symplectic group $\operatorname*{Sp}(n)$ acts transitively on the set
$\operatorname*{Lag}(n)$ of all Lagrangian subspaces of $(\mathbb{R}%
^{2n},\omega)$ (see \cite{Birk} for a proof using symplectic bases). In
particular:%
\[
\text{\textit{The action} }U(n)\times\operatorname*{Lag}(n)\longrightarrow
\operatorname*{Lag}(n)\text{ \textit{is transitive}}%
\]
where $U(n)\subset\operatorname*{Sp}(n)$ is the group of symplectic rotations
(see Appendix A).

Let $\ell_{(\alpha,\beta)}\in\operatorname*{Lag}(n)$ be a Lagrangian
coordinate subspace. It follows that There exists non-unique) $S_{(\alpha
,\beta)},S_{(\alpha,\beta)}^{\prime}\in\operatorname*{Sp}(n)$ such that
\[
\ell_{(\alpha,\beta)}=S_{(\alpha,\beta)}\ell_{X}=S_{(\alpha,\beta)}^{\prime
}\ell_{P}%
\]
(notice that we can take $S_{(\alpha,\beta)}^{\prime}=S_{(\alpha,\beta)}J$).

\section*{APPENDIX\ C. $\hbar$-Polar Duality}

Let $X\subset\mathbb{R}_{x}^{n}$ be a convex body: $X$ is compact\ and convex,
and has non-empty interior $\operatorname*{int}(X)$. If $0\in
\operatorname*{int}(X)$ we define the $\hbar$-polar dual $X^{\hslash}%
\subset\mathbb{R}_{p}^{n}$ of $X$ by
\begin{equation}
X^{\hslash}=\{p\in\mathbb{R}^{m}:\sup\nolimits_{x\in X}(p\cdot x)\leq\hbar\}
\end{equation}
where $\hbar$ is a positive constant (we have $X^{\hslash}=\hbar X^{o}$ where
$X^{o}$ is the traditional polar dual from convex geometry). The following
properties of polar duality are obvious:

\begin{itemize}
\item $(X^{\hslash})^{\hbar}=X$ (reflexivity) and $X\subset Y\Longrightarrow
Y^{\hslash}\subset X^{\hslash}$ (anti-monotonicity),

\item For all $L\in GL(n,\mathbb{R})$:%
\begin{equation}
(LX)^{\hbar}=(L^{T})^{-1}X^{\hslash} \label{scaling}%
\end{equation}
(scaling property). In particular $(\lambda X)^{\hbar}=\lambda^{-1}X^{\hslash
}$ for all $\lambda\in\mathbb{R}$, $\lambda\neq0$.
\end{itemize}

We can view $X$ and $X^{\hslash}$ as subsets of phase space by the
identifications $\mathbb{R}_{x}^{n}\equiv\mathbb{R}_{x}^{n}\times0$ and
$\mathbb{R}_{p}^{n}\equiv0\times\mathbb{R}_{p}^{n}$. Writing $\ell
_{X}=\mathbb{R}_{x}^{n}\times0$ and $\ell_{P}=0\times\mathbb{R}_{p}^{n}$ the
transformation $X\longrightarrow X^{\hslash}$ is a mapping $\ell
_{X}\longrightarrow\ell_{P}$. With this interpretation formula (\ref{scaling})
can be rewritten in symplectic form as%
\begin{equation}
(M_{L^{-1}}X)^{\hbar}=M_{L^{T}}X^{\hslash} \label{ML}%
\end{equation}
where $M_{L}=%
\begin{pmatrix}
L^{-1} & 0\\
0 & L^{T}%
\end{pmatrix}
$ is in $\operatorname*{Sp}(n)$. Notice that $M_{L}:\ell_{X}\longrightarrow
\ell_{X}$ and $M_{L}:\ell_{P}\longrightarrow\ell_{P}$.

Suppose now that $X$ is an ellipsoid centered at the origin:
\begin{equation}
X=\{x\in\mathbb{R}_{x}^{n}:Ax\cdot x\leq\hbar\}
\end{equation}
where $A\in\operatorname*{Sym}_{++}(n,\mathbb{R})$. The polar dual
$X^{\hslash}$ of $X$ is the ellipsoid%
\begin{equation}
X^{\hslash}=\{p\in\mathbb{R}_{p}^{n}:A^{-1}p\cdot p\leq\hbar\}.
\end{equation}
In particular the polar dual of the ball $B_{X}^{n}(\sqrt{\hbar}%
)=\{x:|x|\leq\sqrt{\hbar}\}$ is
\[
(B_{X}^{n}(\sqrt{\hbar}))^{\hbar}=B_{P}^{n}(\sqrt{\hbar}).
\]

\begin{acknowledgement}
This work has been financed by the Grant P 33447 of the Austrian Research
Foundation FWF.
\end{acknowledgement}

\end{document}